\RequirePackage{amsmath} 
\documentclass{llncs}
\usepackage{algorithm,amsfonts,amssymb,bm,verbatim}
\usepackage{tabularx,arydshln}
\usepackage{cellspace,chngcntr,colortbl,epsfig,graphicx,mathtools,xr-hyper,hyperref,multirow,tikz,xcolor}
\usepackage[misc]{ifsym}
\usepackage[all]{xy}
\usepackage[noend]{algpseudocode}
\usepackage[normalem]{ulem} 

\renewcommand{\v}[1]{{\bm #1}}
\renewcommand{\qed}{\hfill\square}
\spnewtheorem{assumption}[theorem]{Assumption}{\bfseries}{\itshape}
\spnewtheorem{prop}[theorem]{Proposition}{\bfseries}{\itshape}
\spnewtheorem{rem}[theorem]{Remark}{\bfseries}{\itshape}
\newcommand{\keywords}[1]{\par\addvspace\baselineskip\noindent\keywordname\enspace\ignorespaces#1}

\DeclarePairedDelimiter\bra{\langle}{\rvert}
\DeclarePairedDelimiter\ket{\lvert}{\rangle}
\DeclarePairedDelimiterX\braket[2]{\langle}{\rangle}{#1 \delimsize\vert #2}

\algrenewcommand\algorithmicrequire{\textbf{Input:}}
\algrenewcommand\algorithmicensure{\textbf{Output:}}

\graphicspath{{../figures/}{./figures/}{../../figures/}}
\makeatletter
\def\input@path{{./texts/}}
\makeatother 

\begin{document}

\title{Reachability Constraints in Variational Quantum Circuits: \\ Optimization within Polynomial Group Module}

\author{Yun-\!Tak\! Oh$^{(\text{\Letter})}$\!\and\!\!\! Dongsoo\! Lee \and\!Jungyoul\! Park  \and\!Kyung\! Chul\! Jeong \!\!\and\! Panjin\;Kim}
\institute{The Affiliated Institute of ETRI, Daejeon 34044, Korea \newline ytak0105@nsr.re.kr}

\maketitle
\begin{abstract}
This work identifies a necessary condition for any variational quantum approach to reach the exact ground state. Briefly, the norms of the projections of the input and the ground state onto each group module must match, implying that module weights of the solution state have to be known in advance in order to reach the exact ground state.
An exemplary case is provided by matchgate circuits applied to problems whose solutions are classical bit strings, since all computational basis states share the same module-wise weights. Combined with the known classical simulability of quantum circuits for which observables lie in a small linear subspace, this implies that certain problems admit a classical surrogate for exact solution with each step taking $O(n^5)$ time.
The Maximum Cut problem serves as an illustrative example.
\keywords{Variational Quantum Circuit, Classical Optimization, NISQ}
\end{abstract}

\section{Introduction}\label{sec:intro}
Studies of the barren plateau phenomenon in quantum machine learning have greatly improved our understanding of the nature of variational quantum computing~\cite{BP18,BP23,BP24a,BP24b,BP25,BP25review}.
Specifically, recent developments in the dynamical Lie algebraic theory of barren plateaus not only address questions such as why and when barren plateaus occur, but also provide guidance on how to circumvent them, paving the way to potentially practical applications~\cite{gsim,QRENN}.

Variational quantum algorithms are generally formulated as quantum-classical hybrid optimization procedures~\cite{QAOA,VQE}. It typically involves iteratively optimizing the parameters of a quantum circuit, which implements a unitary transformation on a $2^n$-dimensional Hilbert space, where $n$ is the number of qubits. However, the seeming advantage of searching an exponentially large space is not always beneficial, as it may be impeded by phenomena such as barren plateaus or poor local minima~\cite{BP18,AK22}. From this perspective, a promising direction for addressing some of these challenges comes from Lie algebraic theory. Under the induced action, the operator space decomposes as a direct sum of smaller invariant subspaces~\cite{BP23}. Indeed, it is known how a unitary operation on a $2^n$-dimensional Hilbert space is represented on these subspaces, and moreover, there is a strong connection between classical simulability and the absence of barren plateaus~\cite{BP25,gsim}.

Building on this line of research, this work investigates how the action of a Lie group relates the input state, the observable, and the associated ground state. The primary focus is on variational algorithms (rather than machine learning) in which the Hamiltonian is fixed and the objective is to prepare its ground state. The results are summarized as follows:
\begin{itemize}
  \item[$\bullet$] A lower bound on the gap between the expectation values of the ground state and the smallest one attainable by any variational method is given. 
  The size of the gap depends on the weight distribution of an input state and the ground state over group modules.
  \item[$\bullet$] For matchgate circuits~\cite{matchgate,JM08}, \textit{all computational basis states} share the same weight distribution. This observation implies that, for certain Hamiltonians whose ground states are not explicitly known but are known to lie within the computational basis, the aforementioned gap can be reduced to zero.  
  \item[$\bullet$] Refining the existing results in Ref.~\cite{gsim}, an effective representation of the cost function is derived for Hamiltonians supported on polynomial subspaces.
  \item[$\bullet$]  The proposed method is numerically tested on the maximum cut (MaxCut) problem, for which the exact solution is encoded in a ground state that is a computational basis state and therefore satisfies the reachability constraints. Both the expectation value and the gradient can be computed within $O(n^5)$ time.
\end{itemize}
Quoting Goh et al., who introduced a Lie-algebraic classical simulation~\cite{gsim}, on page 7 of the reference, the authors summarized their work as
\begin{itemize}
\item[(1)] Wick-based simulation: general observable, special state;
\item[(2)] $\mathfrak{g}$-sim-based simulation: special observable, general state.
\end{itemize}
The result of this work fits naturally into this framework as the third case,
\begin{itemize}
\item[(3)] this work: special observable, special state,
\end{itemize}
although the meaning of `special state'\label{ss} here differs from that in (1).
In this work, a \textit{special state} refers to a class of states that satisfy a necessary condition for exact optimization in quantum variational methods (see Section~\ref{sec:main}).

\begin{figure*}[htbp]
    \centering
    \includegraphics[width=\textwidth]{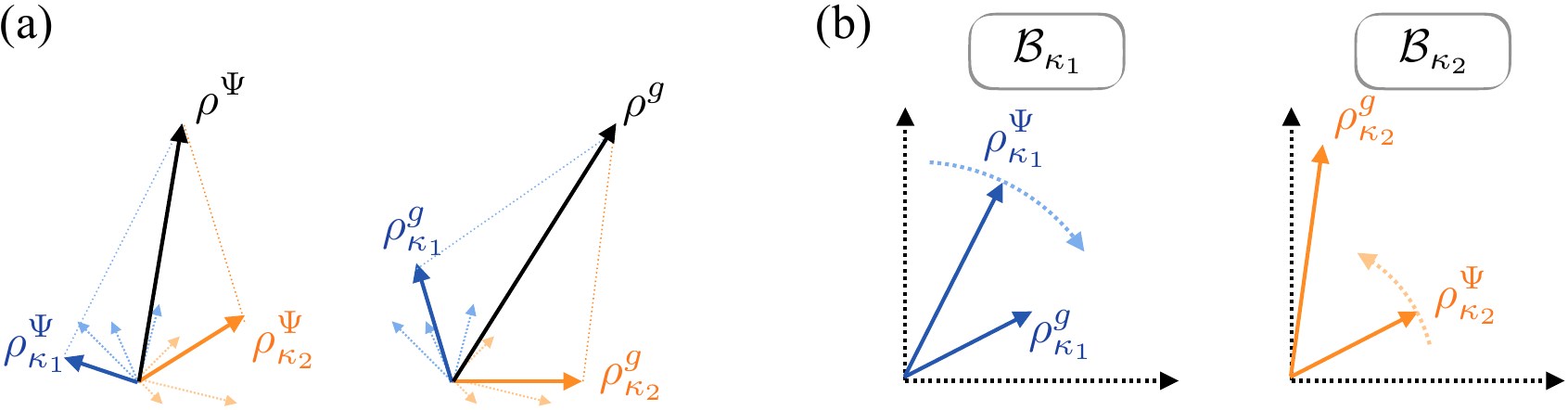}
    \caption{
    (a) Schematic projection of \(\rho^{\mathrm \Psi}\) and \(\rho^g\) onto two representative modules \(\mathcal{B}_\kappa\). 
    The projected components are shown as vectors in the corresponding subspaces (blue and orange). 
    (b) Within each module, unitary transformations act as rotations that preserve the magnitude of the projected component. 
    If $\| \rho^{\mathrm \Psi}_\kappa \|_\mathrm{HS} \neq \| \rho^g_\kappa \|_\mathrm{HS}$, the two vectors cannot be matched by such rotations.
    }
    \label{fig:schematic}
\end{figure*}

It is not claimed in this paper that the proposed algorithm solves NP-complete problems since how the number of iterations scales with the problem size is unclear.
The numerical results roughly suggests that the success probability decreases rather rapidly as $n$ grows.
One of the promising directions for future research is a rigorous and thorough convergence analysis of the algorithm.
\section{Background}\label{sec:background}
The algebraic structure of parameterized quantum circuits and variational approaches is briefly introduced in this section.
Readers are assumed to be familiar with basic notions in quantum information such as Pauli operators, logic circits, and so on~\cite{QCbook}.

\subsection{Parameterized quantum circuit and group module}

Let $\mathcal{H}$ denote the $2^{n}$-dimensional Hilbert space of an $n$-qubit system, and let
\begin{align*}
\mathcal{B}(\mathcal{H})=\{A:\mathcal{H}\rightarrow\mathcal{H}\}
\end{align*}
be the vector space of all bounded linear operators on \(\mathcal{H}\).  We equip \(\mathcal{B}(\mathcal{H})\) with the Hilbert--Schmidt inner product and norm,
\begin{align*}
    \langle A,B\rangle_{\mathrm{HS}}
=
    \operatorname{Tr} [A^{\dagger}B],
\qquad
    \| A \|_{\rm HS}
= 
    \langle A,A\rangle_{\mathrm{HS}},
\qquad
    A,B\in\mathcal{B}(\mathcal{H}) ,
\end{align*}
which induces an orthonormal operator basis $\{b_{i}\}_{i=1}^{2^{2n}}\subset\mathcal{B}(\mathcal{H})$ satisfying
\begin{align*}
\langle b_i,b_j\rangle_{\mathrm{HS}}=\delta_{i,j}.
\end{align*}
The subscript ${\rm HS}$ will be omitted hereafter.

Let $U(\boldsymbol{\theta})$ be a parameterized quantum circuit with parameters $\boldsymbol{\theta}=(\theta_1,\theta_2,\ldots)$. The unitary group  $\mathcal{U}$ generated by the circuit family are characterized by the so-called dynamical Lie algebra $\mathfrak{g}$~\cite{intro-quantum-dynamics,zeier11}. The group $\mathcal{U}$ acts on $\mathcal{B}(\mathcal{H})$ by conjugation, $A\mapsto UAU^{\dagger}$, $U\in\mathcal{U}$. Since this action is unitary with respect to $\langle\cdot,\cdot\rangle$ and $\mathcal{U}$ is compact, $\mathcal{B}(\mathcal{H})$ decomposes into a direct sum of $\mathcal{U}$-invariant subspaces (group modules),
\begin{align}
\mathcal{B}(\mathcal{H})=\bigoplus_{\kappa}\mathcal{B}_{\kappa}.
\label{eq:operator_space_decomp}
\end{align}
Each module $\mathcal{B}_{\kappa}$ satisfies
\begin{align*}
\forall\,U\in\mathcal{U},\;\forall\,A\in\mathcal{B}_{\kappa}:\qquad
U A U^{\dagger}\in\mathcal{B}_{\kappa}.
\end{align*}
Consequently, any operator $O\in\mathcal{B}(\mathcal{H})$ admits the unique decomposition
\begin{align*}
O = \sum_{\kappa} O_{\kappa},
\qquad
O_{\kappa}\in\mathcal{B}_{\kappa},
\end{align*}
where $O_{\kappa}$ is the orthogonal projection onto $\mathcal{B}_{\kappa}$.  
If $\{b^{l}_{\kappa}\}_{l=1}^{d_{\kappa}} $ is an orthonormal basis of $\mathcal{B}_{\kappa}$ and $d_{\kappa}=\dim (\mathcal{B}_{\kappa})$, then
\begin{align*}
O_{\kappa}= \sum_{l=1}^{d_{\kappa}}
\operatorname{Tr}\!\big[(b^{l}_{\kappa})^{\dagger} O \big]\; b^{l}_{\kappa}.
\end{align*}

\subsection{Variational optimization and barren plateau}

Variational quantum algorithms optimize a parameterized circuit $U(\boldsymbol{\theta})$ by minimizing a cost function defined from measurement outcomes.  A common setting prepares a state
\begin{align*}
    \rho^\Psi(\boldsymbol{\theta})
=
    U(\boldsymbol{\theta})\,\rho^\mathrm{i}\,U(\boldsymbol{\theta})^{\dagger} ,
\end{align*}
where $\rho^\mathrm{i}$ is the initial state, and then minimizes
\begin{align}
    C(\boldsymbol{\theta})
    :=
    \operatorname{Tr}\bigl[\rho^\Psi(\boldsymbol{\theta})\, H \bigr] , 
\label{eq:cost-fn}
\end{align}
where \(H \in \mathcal{B}(\mathcal{H})\) is an observable, commonly referred to as a Hamiltonian. In many cases, the problem of finding its ground state can be reduced to solving a corresponding classical optimization problem. In applications such as the variational quantum eigensolver or the quantum approximate optimization algorithm, $H$ encodes an objective Hamiltonian such that minimizing the energy yields an (approximate) solution~\cite{QAOA,VQE}. In quantum machine learning, $H$ defines a task-specific loss to be minimized~\cite{QNN18a,QNN18b,QNN19}.

The group module decomposition in Eq.~\eqref{eq:operator_space_decomp} provides a rigorous framework for analyzing the optimization landscape. We decompose the initial state $\rho^{\textrm i}$ and the observable $H$ into their module components,
\begin{align*}
\rho^{\textrm{i}}= \sum_{\kappa}\rho^{\textrm{i}}_{\kappa}, \qquad
H= \sum_{\kappa} H_{\kappa}.
\end{align*}%
Projection onto each module is described in the next section. Because the circuit unitary $U(\boldsymbol{\theta})$ does not mix different modules, the evolved state within a given module is obtained by applying the unitary to that component,
\begin{align*}
\rho_{\kappa}^\Psi(\boldsymbol{\theta}) = U(\boldsymbol{\theta})\,\rho^{\textrm{i}}_{\kappa}\,U^{\dagger}(\boldsymbol{\theta}) .
\end{align*}
Consequently, the cost function can be expressed as a sum over module contributions,
\begin{align}
C(\boldsymbol{\theta}) = \sum_{\kappa} C_{\kappa}(\boldsymbol{\theta})
= \sum_{\kappa} \operatorname{Tr}\!\bigl[\rho_{\kappa}^\Psi(\boldsymbol{\theta})\,H_\kappa  \bigr] 
= \sum_{\kappa} \operatorname{Tr}\!\bigl[
U(\boldsymbol{\theta})\,\rho^{\textrm{i}}_{\kappa}\,U^{\dagger}(\boldsymbol{\theta})
\,H_\kappa
\bigr] 
.
\label{eq:group-mod-loss}
\end{align}
When a Hamiltonian \(H\) has support only on a few modules, the effective dimension that governs $C(\boldsymbol{\theta})$ can be much smaller than $\dim \bigl(\mathcal{B}(\mathcal{H})\bigr)$. This formulation will be employed in Section~\ref{sec:main} to derive effective expressions for the cost function and to discuss the implications for gradient concentration.

Equation~\eqref{eq:group-mod-loss} implies each group module behaves independently, allowing barren plateau analysis to be carried out on a module‑by‑module basis. In Ref.~\cite{BP23}, it is shown that each module contributes independently to the variance of the loss gradient,
\begin{align*}
\operatorname{Var}_{\boldsymbol{\theta}}\!\bigl[\partial_{\theta} C(\boldsymbol{\theta})\bigr]
\approx
\sum_{\kappa} \frac{\|\rho^{\textrm{i}}_{\kappa}\|\,\|H_{\kappa}\|}{\dim (\mathcal{B}_{\kappa})} .
\end{align*}
If the observable is supported only on a few modules whose dimensions scale polynomially with the system size, the barren plateau phenomenon is expected to be absent during optimization.

\section{Reachability constraints}\label{sec:main}

Consider a family of circuits whose conjugation action preserves the group module decomposition \(\mathcal{B}(\mathcal{H}) = \bigoplus_{\kappa} \mathcal{B}_{\kappa}\). Let \(C(\v \theta)\), defined in Eq.~\eqref{eq:cost-fn}, denote the cost function to be minimized and let $\rho^{\mathrm i}$ be the initial state. For simplicity, we assume that the Hamiltonian \(H\) admits a unique ground state. We denote its smallest eigenvalue by \(E_g\) and the corresponding ground state by \(\rho^g\). For any density matrix \(\rho\), let \(\rho_\kappa\) denote its orthogonal projection onto \(\mathcal{B}_\kappa\).

\begin{theorem}[Reachability constraints]
Consider a variational algorithm in the above setting.
Depending on the weight distribution of the initial state and the ground state over the modules, the following statements hold:

\begin{itemize}
\item[$\bullet$] If \(\bigl\| \rho^{\mathrm{i}}_{\kappa} \bigr\| = \bigl\| \rho^{g}_{\kappa} \bigr\|\) for all \(\kappa\), then there can exist \(U(\v \theta)\) such that \(C(\v \theta) = E_g\).

\vspace{3mm}

\item[$\bullet$] If \(\bigl\| \rho^\mathrm{i}_{\kappa} \bigr\| \ne \bigl\| \rho^{g}_{\kappa} \bigr\|\) for some \(\kappa\), then for all \(U(\v \theta)\),
\begin{align*}
C(\v \theta) \ge E_g \pi_{g} + E_1 (1- \pi_{g} ) ,
\end{align*}
where \(E_1\) is the first excited energy and
\(
\pi_{g} = \sum_\kappa \| \rho^g_\kappa \| \cdot \| \rho^\mathrm{i}_\kappa \|.
\)
\end{itemize}
\label{thm:reacheability}
\end{theorem}

\begin{proof}
By construction, conjugation by the circuit preserves each module:
for all \(A \in \mathcal{B}_\kappa\),
\(
U(\v \theta) \,A \, U(\v \theta)^{\dagger} \in \mathcal{B}_{\kappa}.
\)
Furthermore, conjugation is an isometry with respect to the Hilbert--Schmidt inner product,
\begin{align*}
\langle U(\v \theta) A U(\v \theta)^{\dagger}, U(\v \theta) B U(\v \theta)^{\dagger} \rangle
=
\langle A, B \rangle .
\end{align*}
Therefore, the norm of each projected component of the state is preserved,
\begin{align*}
\bigl\|
\bigl(U(\v \theta) \rho^\textrm{i} U(\v \theta)^{\dagger} \bigr)_{\kappa}
\bigr\|
=
\bigl\|
\rho^\textrm{i}_{\kappa}
\bigr\|.
\end{align*}

Assume there exists a circuit such that
\(
U(\v \theta) \rho^\textrm{i} U(\v \theta)^{\dagger} = \rho^{g}.
\)
Decomposing both states into module components,
\begin{align*}
\rho^\textrm{i} = \sum_{\kappa} \rho^\textrm{i}_{\kappa},
\qquad
\rho^g = \sum_{\kappa} \rho^{g}_{\kappa},
\end{align*}
the invariance above implies a necessary condition
\begin{align}
\bigl\| \rho^\textrm{i}_{\kappa} \bigr\|
=
\bigl\| \rho^g_{\kappa} \bigr\|
\qquad \forall \kappa .
\label{eq:module_weight_condition}
\end{align}

We now consider the case in which Eq.~\eqref{eq:module_weight_condition} is violated. In this situation, the achievable overlap with the ground state is fundamentally limited. In particular, by the Cauchy--Bunyakovsky--Schwarz inequality,
\begin{align*}
\operatorname{Tr}[\rho^g \rho^{\Psi}]
=
\sum_{\kappa} \operatorname{Tr}[\rho_\kappa^g \rho_\kappa^{\Psi}]
\le
\sum_{\kappa}
\| \rho_\kappa^g \| \cdot \| \rho_\kappa^{\Psi} \|
=
\sum_{\kappa}
\| \rho_\kappa^g \| \cdot \| \rho_\kappa^{\mathrm i} \|.
\end{align*}
Thus, no reachable state \(\rho^{\Psi}\) can attain full overlap with \(\rho^g\).

Writing the Hamiltonian as
\begin{align*}
H = E_g \rho^g + \sum_{\lambda \ne g} E_\lambda \rho^\lambda ,
\end{align*}
the expectation value reads
\begin{align}
\operatorname{Tr}[\rho^{\Psi}\,H ]
&=
E_g \operatorname{Tr}[\rho^{\Psi} \rho^g ]
+
\sum_{\lambda \ne g} E_\lambda \operatorname{Tr}[\rho^{\Psi} \rho^\lambda ] .
\label{eq:lower-bound-1}
\end{align}
Substituting the above inequailty into Eq.~\eqref{eq:lower-bound-1}, we obtain
\begin{align*}
\operatorname{Tr}[\rho^{\Psi}H]
\ge
E_g \pi_g
+
\sum_{\lambda \ne g} E_\lambda \operatorname{Tr}[\rho^{\Psi} \rho^\lambda]
\ge
E_g \pi_g + E_1 (1-\pi_g).
\end{align*}
Here, \(E_1\) denotes the energy of the first excited state above the ground state.
\(\qed\)
\end{proof}

Therefore, the equality of module-wise Hilbert--Schmidt norms constitutes a reachability constraint imposed by the circuit structure. If the module weights differ, no admissible unitary can transform the initial state into the ground state.

When the module weight distribution of the target state is known, one can construct an initial state with an identical distribution. Such a tailored initial state, referred to as a \textit{special state} (introduced in the third case in Section~\ref{sec:intro}), can enhance optimization performance in certain settings. We emphasize that the notion of a special state in this work differs from that in Ref.~\cite{BP25}, where the term is used to denote states that are easy to simulate. In contrast, here it refers to states that satisfy the reachability condition and therefore admit the possibility of being fully optimized.

\section{Classical optimization by effective representation}

Suppose \(H\) belongs to a single group module \(\mathcal{B}_{\kappa_p}\subset\mathcal{B}(\mathcal{H})\) such that
\begin{align*}
\dim(\mathcal{B}_{\kappa_p}) = \operatorname{poly}(n),
\qquad
H \in \mathcal{B}_{\kappa_p}.
\end{align*}
From Eq.\,\eqref{eq:group-mod-loss}, one can directly show that the cost function depends only on the \(\kappa_p\)-component,
\begin{align}
C(\boldsymbol{\theta})
=
\operatorname{Tr}[\rho^{\Psi} (\boldsymbol{\theta}) \, H]
=
\operatorname{Tr}[\rho^{\Psi}_{\kappa_p} (\boldsymbol{\theta})\,  H]
=
\operatorname{Tr}
\bigl[
H
\,
U(\boldsymbol{\theta})\,\rho^{\textrm{i}}_{\kappa_p}\,U^{\dagger}(\boldsymbol{\theta})
\bigr] .
\end{align}
Therefore it suffices to search the ground state within \(\mathcal{B}_{\kappa_p}\) while discrepancies in other modules do not affect the objective. 

We now formulate a classical optimization framework based on the effective representation introduced in~\cite{BP25}.
Restricting to \(\mathcal{B}_{\kappa_p}\), we expand the projected initial state and Hamiltonian in an orthonormal operator basis \(\{b^{\,l}_{\kappa_p}\}\),
\begin{align*}
\rho^\mathrm{i}_{\kappa_p}
=
\sum_{l_1=1}^{d_{\kappa_p}}
\varphi^{\,l_1}_{\kappa_p} b^{\,l_1}_{\kappa_p},
\qquad 
H =
\sum_{l_2=1}^{d_{\kappa_p}} \eta^{\,l_2}_{\kappa_p} b^{\,l_2}_{\kappa_p}.
\end{align*}
Because the circuit preserves the module structure, its conjugation action induces a linear transformation within \(\mathcal{B}_{\kappa_p}\),
\begin{align*}
U(\boldsymbol{\theta}) b^{\,l}_{\kappa_p} U^\dagger(\boldsymbol{\theta})
=
\sum_{l'}
\bigl[\mathcal{M}(\boldsymbol{\theta})\bigr]_{l,l'}\, b^{\,l'}_{\kappa_p}.
\end{align*}
The matrix elements \(\bigl[\mathcal{M}(\boldsymbol{\theta})\bigr]_{l,l'}\) are obtained via the Hilbert--Schmidt inner product,
\begin{align}
\bigl[\mathcal{M}(\boldsymbol{\theta})\bigr]_{l,l'}
=
\operatorname{Tr}
\!\left[
b^{\,l'}_{\kappa_p}
\left(
U(\boldsymbol{\theta}) b^{\,l}_{\kappa_p} U^\dagger(\boldsymbol{\theta})
\right)
\right].
\label{eq:eff-U}
\end{align}
Substituting these expansions into the cost function yields
\begin{align}
C(\boldsymbol{\theta})
=
\sum_{l_1,l_2}
\varphi^{\,l_1}_{\kappa_p}
\bigl[\mathcal{M}(\boldsymbol{\theta})\bigr]_{l_1,l_2}
\eta^{\,l_2}_{\kappa_p}.
\label{eq:eff-expectation}
\end{align}
All quantities involved scale polynomially in \(n\).

When the circuit generators, the target observable, and the ground state satisfy the aforementioned structural conditions, Eq.~\eqref{eq:eff-expectation} can be computed efficiently on a classical computer enabling a search for the exact ground state. Such an observable fall into the class of \textit{special observable} appeared in the introduction.

\section{Application}
\subsection{Matchgate circuit and MaxCut Hamiltonian}
As an example, we examine the variational quantum eigensolver for the MaxCut problem on 3-regular graphs via the matchgate circuit. The $n$-qubit matchgate circuit is generated by
\begin{align}
\mathcal{G}
=
\{Z_i\}_{i=1}^{n}
\cup
\{X_i X_{i+1}\}_{i=1}^{n-1},
\end{align}
and its dynamical Lie algebra consists of quadratic Majorana operators. Here, \(X_i,Y_i,Z_i\) are the single qubit Pauli operators acting on site \(i\). The operator space decomposes as
\begin{align*}
\mathcal{B}(\mathcal{H})
=
\bigoplus_{\kappa=0}^{2n} \mathcal{B}_{\kappa},
\qquad
\dim (\mathcal{B}_{\kappa}) = \binom{2n}{\kappa}.
\end{align*}
Introducing \(2n\) Majorana operators \( \{c_\mu^{X}, c_\mu^{Y}\}_{\mu=1}^{n}\), the \(\kappa\)-th group module is spanned by basis operators constructed from multiplying \(\kappa\) Majorana operators (see Appendix~\ref{appsec:majorana} for details). 

The MaxCut problem on a graph \(G=(V,E)\) can be encoded in the Hamiltonian
\begin{align*}
H_{\mathrm{MaxCut}} = \sum_{( i,j ) \in E} Z_i Z_j .
\end{align*}
Each term \(Z_i Z_j\) can be expressed as a product of four Majorana operators and therefore belongs to \(\mathcal{B}_{4}\)\footnote{Therefore any Ising type Hamiltonian belongs to $\mathcal{B}_4$}. Consequently,
\begin{align*}
H_{\mathrm{MaxCut}} \in \mathcal{B}_{4},
\qquad
\dim (\mathcal{B}_{4}) = \binom{2n}{4} = O(n^4).
\label{eq:maxcut-b4}
\end{align*}

It is worth noting that in this setup, one can satisfy the reachability constraint without knowing the ground state.

\begin{rem}
For computational basis states, the Hilbert--Schmidt norm of the projection onto $\mathcal{B}_\kappa$ are the same. In other words, for any computational states $\rho^a$ and $\rho^b$,
\begin{align}
\bigl\| \rho^{a}_{\kappa} \bigr\|
=
\bigl\| \rho^{b}_{\kappa} \bigr\|
\text{~~for all } \kappa .
\end{align}
\end{rem}
Since the ground state of the MaxCut problem is a computational state, the constraint is satisfied as long as the initial state is also a computational state.

\begin{prop}\label{prop:complexity}
The classical time and memory complexity of computing the expectation value of $H_{\mathrm{MaxCut}}$ is bounded by $O(n^5)$.
\end{prop}
The proof is given in Appendix~\ref{appsec:complexity}.
Once the optimization converges to the ground state, the solution (maximum cut partition) can be extracted by computing $\langle Z_i \rangle$ for all $i$.

\subsection{Numerical results}
Using the method described in Appendix~\ref{sec:app-matchgate-group-module}, we compute Eq.~\eqref{eq:eff-expectation} classically with its time and memory costs being bounded by \(O(n^{5})\) per single-shot evaluation (i.e., for a fixed set of parameters).

We compare the proposed method (projected method hereafter) against a conventional classical quantum circuit simulator implemented using the Pennylane library. Parameters are initialized randomly. Since matchgate circuits preserve the parity, we consider the zero string state \(\ket{\psi^{0} } = \bigotimes_{i} \ket{ 0}_{i}, \) and the single-qubit flipped state \(X_{1}|\psi^{0}\rangle\) as the initial state. In optimization, the gradient of the expectation value is computed via back-propagation. Further experimental details are given in Appendix~\ref{appsec:experiment}.

\begin{figure*}[htbp]
    \centering
    \includegraphics[width=0.9\textwidth]{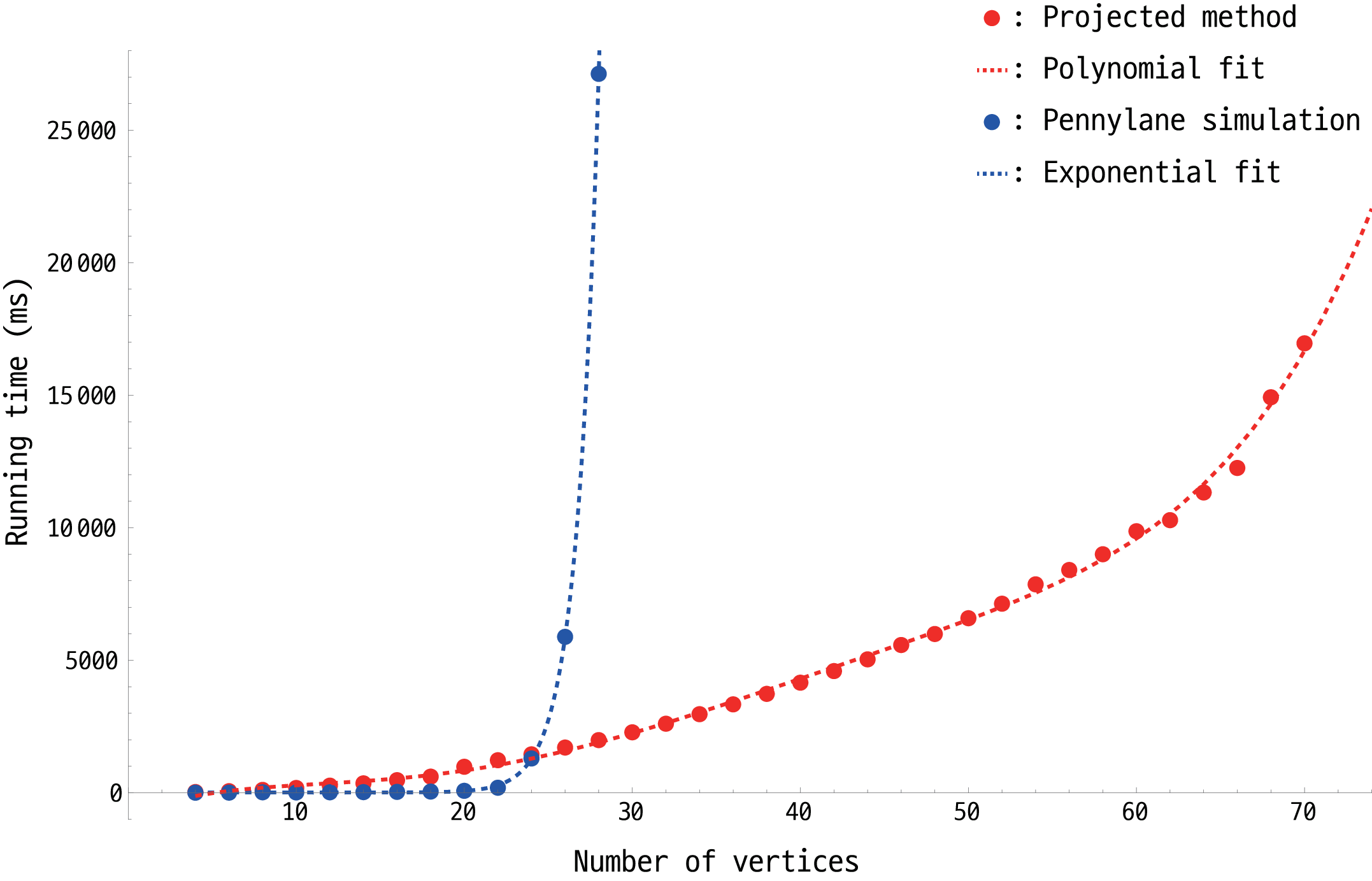}
    \caption{
    Running time as a function of $n$ (the number of vertices) for evaluating the expectation value of $H_{\mathrm{MaxCut}}$ using a matchgate circuit. 
    In Pennylane simulation, $n$ is the number of qubits.
    Each marker represents an average value of 100 trials.
    }
    \label{fig:time-scaling}
\end{figure*}

Figure~\ref{fig:time-scaling} shows the running time for a single-shot evaluation as a function of $n$. Each Red (blue) marker is the average running time over 100 trials for the projected method (Pennylane simulation).
The red (blue) dashed line is a polynomial (exponential) fit to the respective data points.
While the theoretical time cost is $O(n^5)$ for the projected method, the polynomial fitting curve we obtained (up to the fifth order in $n$) reads $0.000115186 n^5 - 0.0183293 n^4 + 1.0483 n^3 - 23.1588 n^2 + 255.757 n - 835.677$~(\textrm{ms}).
We suspect that the reason for the small coefficient of the quintic term originates from the parallelization inherent in GPU execution.
For $n < 64$, we were not able to observe the full utilization of a single GPU.\footnote{The limit on the largest graph size in the figure ($n =70$) is set by the memory limit of a single GPU. For the record, we could tackle $n=84$ by using four GPUs.}
Once the GPU utilization hits $100\%$ at $n=64$, from that on, it is shown in the figure that the behavior is different from $n<64$ which should read $O(n^5)$.

To test if the variational approach can reach the exact ground state when the constraint is satisfied,  for each system size \(n\) ($=|V|$), we generate 10 random instances of 3-regular graphs and attempt to solve them. For each instance, we perform up to 10 independent trials with random parameter initialization for each parity sector (even and odd). The success rates of finding the optimal solution up to \(n=60\) are summarized in Table~\ref{tab:success-rate}. While the optimal solution is consistently obtained for small system sizes, it decreases rather rapidly as \(n\) increases. A systematic analysis of how the required number of trials scales with system size remains open.

\begin{table}[htbp]
  \caption{Success rate of reaching the optimal state $\ket{\psi}$ such that $E_g = \bra{\psi} H_{\mathrm{MaxCut}} \ket{\psi}$.
  For each $n$, 10 instances are generated, and for each instance, up to 10 trials are attempted.
  It is counted as success if the optimal state is found within 10 trials for an instance.
  }
  \centering
  \begin{tabular}{c | >{\centering}p{0.5cm} >{\centering}p{0.5cm} >{\centering}p{0.5cm} >{\centering}p{0.7cm} >{\centering}p{0.7cm} >{\centering}p{0.7cm} >{\centering}p{0.7cm} >{\centering}p{0.7cm} >{\centering}p{0.7cm} >{\centering}p{0.7cm} >{\centering}p{0.7cm} >{\centering}p{0.7cm} >{\centering}p{0.7cm} >{\centering}p{0.7cm} >{\centering}p{0.7cm}}
     $n$&$4$&$8$&$12$&$16$&$20$ &$24$ &$28$ &$32$ &$36$ &$40$ &$44$ &$48$ &$52$ &$56$ &$60$ \tabularnewline \hline
    rate&$1$&$1$&$1$ &$1$ &$0.9$&$0.9$&$0.7$&$0.9$&$0.6$&$0.6$&$0.5$&$0.4$&$0.4$&$0.3$&$0.4$   \end{tabular}
\label{tab:success-rate}
\end{table}

Additional data including convergence curves for a random 3-regular graph and a benchmark result on $p=0.5$ Erd\H{o}s-R\'{e}nyi random graph~\cite{biqmac} is given in Appendix~\ref{appsec:convergence}.

\section{Conclusion}
A projected formulation of variational quantum optimization, grounded in the group module decomposition of operator space and Hilbert--Schmidt geometry, reveals a structural reachability constraint imposed by circuit symmetry.
For circuit families whose conjugation action preserves the decomposition 
\(
\mathcal{B}(\mathcal{H})=\bigoplus_{\kappa}\mathcal{B}_{\kappa},
\)
the invariance of each group module implies conservation of the Hilbert--Schmidt norm of the projected density operator within every module, yielding a necessary condition for reachability: a trial state can be mapped to a target state only if their density operators have identical module-wise Hilbert--Schmidt norms, which is independent of the problem Hamiltonian.


An intriguing structural observation is that all computational basis states share identical module weight distributions when matchgate circuit is considered.
Consequently, when the ground state of a problem Hamiltonian is a computational basis state, the reacheability condition is automatically satisfied for example by the reference state \( \ket{0 \ldots 0} \).  
Within the projected picture applied to the MaxCut problem, classical optimization of an effective matrix acting on
\(H_\mathrm{MaxCut} \in \mathcal{B}_{4}\)
sometimes yields the optimal circuit parameters in numerical experiments.

We emphasise that this does not render NP-hard problems solvable in polynomial time.  
Matchgate circuits are already known to admit efficient classical descriptions in the free fermionic framework.  
What is different here is the operator-algebraic explanation of reachability and variational reduction, which makes explicit the role of module weight invariance. 



\newpage

\bibliographystyle{splncs04}
\bibliography{reference}

@misc{biqmac,
  author       = {Angelika Wiegele},
  title        = {Biq Mac Library -- A collection of Max-Cut and quadratic 0-1 programming instances of medium size},
  year         = {2007},
  howpublished = {\url{https://biqmac.aau.at/biqmaclib.html}},
  note         = {Accessed: 2026-04-02}
}

@article{zeier11,
title="{Symmetry principles in quantum system theory}",
author={Robert Zeier  and Thomas Schulte-Herbr{\''u}ggen},
journal={Journal of mathematical physics},
volume={52},
number={11},
year={2011},
publisher={API Publishing}
}

@book{intro-quantum-dynamics,
title="{Introduction to quantum control and dynamics}",
author={Domenico d' Alessandro},
year={2021},
publisher={Chapman and hall/CRC}
}

@article{Larocca2023overpara,
    title ="{Theory of overparametrization in quantum neural networks}",
    author={Larocca, Martin and Ju, Nathan and García-Martín, Diego and Coles, Patric J and Cerezo, Marco},
    journal={Nature Computational Science},
    volume={3},
    number={6},
    pages={542--551},
    year={2023},
    publisher={Nature Publishing Group US New York}
}

@book{QCbook,
	title ="{Quantum computation and quantum information: 10th anniversary edition}",
	DOI={10.1017/CBO9780511976667},
	publisher={Cambridge University Press},
	author={Nielsen, Michael A. and Chuang, Isaac L.},
	year={2010}
}

@misc{QAOA,
      title="{A auantum approximate optimization algorithm}",
      author={Edward Farhi and Jeffrey Goldstone and Sam Gutmann},
      year={2014},
      eprint={1411.4028},
      archivePrefix={arXiv},
      primaryClass={quant-ph},
      url={https://arxiv.org/abs/1411.4028},
}

@article{VQE,
author={Peruzzo, Alberto and McClean, Jarrod and Shadbolt, Peter and Yung, Man-Hong and Zhou, Xiao-Qi and Love, Peter J. and Aspuru-Guzik, Al{\'a}n and O'Brien, Jeremy L.},
title="{A variational eigenvalue solver on a photonic quantum processor}",
journal={Nature Communications},
year={2014},
month={Jul},
day={23},
volume={5},
number={1},
pages={4213},
issn={2041-1723},
doi={10.1038/ncomms5213},
url={https://doi.org/10.1038/ncomms5213}
}

@article{BP18,
author={McClean, Jarrod R. and Boixo, Sergio and Smelyanskiy, Vadim N. and Babbush, Ryan and Neven, Hartmut},
title="{Barren plateaus in quantum neural network training landscapes}",
journal={Nature Communications},
year={2018},
month={Nov},
day={16},
volume={9},
number={1},
pages={4812},
issn={2041-1723},
doi={10.1038/s41467-018-07090-4},
url={https://doi.org/10.1038/s41467-018-07090-4}
}

@misc{BP23,
      title="{Showcasing a barren plateau theory beyond the dynamical Lie algebra}",
      author={N. L. Diaz and Diego García-Martín and Sujay Kazi and Martin Larocca and M. Cerezo},
      year={2023},
      eprint={2310.11505},
      archivePrefix={arXiv},
      primaryClass={quant-ph},
      url={https://arxiv.org/abs/2310.11505},
}

@article{BP24a,
author={Fontana, Enrico and Herman, Dylan and Chakrabarti, Shouvanik and Kumar, Niraj and Yalovetzky, Romina and Heredge, Jamie and Sureshbabu, Shree Hari and Pistoia, Marco},
title="{Characterizing barren plateaus in quantum ans{\"a}tze with the adjoint representation}",
journal={Nature Communications},
year={2024},
month={Aug},
day={22},
volume={15},
number={1},
pages={7171},
issn={2041-1723},
doi={10.1038/s41467-024-49910-w},
url={https://doi.org/10.1038/s41467-024-49910-w}
}

@article{BP24b,
author={Ragone, Michael and Bakalov, Bojko N. and Sauvage, Fr{\'e}d{\'e}ric and Kemper, Alexander F. and Ortiz Marrero, Carlos and Larocca, Mart{\'i}n and Cerezo, M.},
title="{A Lie algebraic theory of barren plateaus for deep parameterized quantum circuits}",
journal={Nature Communications},
year={2024},
month={Aug},
day={22},
volume={15},
number={1},
pages={7172},
issn={2041-1723},
doi={10.1038/s41467-024-49909-3},
url={https://doi.org/10.1038/s41467-024-49909-3}
}

@article{BP25,
author={Cerezo, M. and Larocca, Martin and Garc{\'i}a-Mart{\'i}n, Diego and Diaz, N. L. and Braccia, Paolo and Fontana, Enrico and Rudolph, Manuel S. and Bermejo, Pablo and Ijaz, Aroosa and Thanasilp, Supanut and Anschuetz, Eric R. and Holmes, Zo{\"e}},
title="{Does provable absence of barren plateaus imply classical simulability?}",
journal={Nature Communications},
year={2025},
month={Aug},
day={25},
volume={16},
number={1},
pages={7907},
issn={2041-1723},
doi={10.1038/s41467-025-63099-6},
url={https://doi.org/10.1038/s41467-025-63099-6}
}

@article{BP25review,
author={Larocca, Mart{\'i}n and Thanasilp, Supanut and Wang, Samson and Sharma, Kunal and Biamonte, Jacob and Coles, Patrick J. and Cincio, Lukasz and McClean, Jarrod R. and Holmes, Zo{\"e} and Cerezo, M.},
title="{Barren plateaus in variational quantum computing}",
journal={Nature Reviews Physics},
year={2025},
month={Apr},
day={01},
volume={7},
number={4},
pages={174-189},
issn={2522-5820},
doi={10.1038/s42254-025-00813-9},
url={https://doi.org/10.1038/s42254-025-00813-9}
}

@article{gsim,
  title = "{Lie-algebraic classical simulations for quantum computing}",
  author = {Goh, Matthew L. and Larocca, Martin and Cincio, Lukasz and Cerezo, M. and Sauvage, Fr\'ed\'eric},
  journal = {Phys. Rev. Res.},
  volume = {7},
  issue = {3},
  pages = {033266},
  numpages = {32},
  year = {2025},
  month = {Sep},
  publisher = {American Physical Society},
  doi = {10.1103/3y65-f5w6},
  url = {https://link.aps.org/doi/10.1103/3y65-f5w6}
}

@misc{QRENN,
      title="{Quantum recurrent embedding neural network}",
      author={Mingrui Jing and Erdong Huang and Xiao Shi and Shengyu Zhang and Xin Wang},
      year={2025},
      eprint={2506.13185},
      archivePrefix={arXiv},
      primaryClass={quant-ph},
      url={https://arxiv.org/abs/2506.13185},
}

@article{AK22,
author={Anschuetz, Eric R. and Kiani, Bobak T.},
title="{Quantum variational algorithms are swamped with traps}",
journal={Nature Communications},
year={2022},
month={Dec},
day={15},
volume={13},
number={1},
pages={7760},
issn={2041-1723},
doi={10.1038/s41467-022-35364-5},
url={https://doi.org/10.1038/s41467-022-35364-5}
}

@article{QNN19,
author={Havlicek, Vojtech and C{\'o}rcoles, Antonio D. and Temme, Kristan and Harrow, Aram W. and  andala, Abhinav and Chow, Jerry M. and Gambetta, Jay M.},
title="{Supervised learning with quantum-enhanced feature spaces}",
journal={Nature},
year={2019},
month={Mar},
day={01},
volume={567},
number={7747},
pages={209-212},
issn={1476-4687},
doi={10.1038/s41586-019-0980-2},
url={https://doi.org/10.1038/s41586-019-0980-2}
}

@article{QNN18a,
  title = "{Quantum circuit learning}",
  author = {Mitarai, K. and Negoro, M. and Kitagawa, M. and Fujii, K.},
  journal = {Phys. Rev. A},
  volume = {98},
  issue = {3},
  pages = {032309},
  numpages = {6},
  year = {2018},
  month = {Sep},
  publisher = {American Physical Society},
  doi = {10.1103/PhysRevA.98.032309},
  url = {https://link.aps.org/doi/10.1103/PhysRevA.98.032309}
}

@misc{QNN18b,
      title= "{Classification with Quantum Neural Networks on Near Term Processors}", 
      author={Edward Farhi and Hartmut Neven},
      year={2018},
      eprint={1802.06002},
      archivePrefix={arXiv},
      primaryClass={quant-ph},
      url={https://arxiv.org/abs/1802.06002}, 
}

@article{JM08,
    author = {Jozsa, Richard and Miyake, Akimasa},
    title = "{Matchgates and classical simulation of quantum circuits}",
    journal = {Proceedings of the Royal Society A: Mathematical, Physical and Engineering Sciences},
    volume = {464},
    number = {2100},
    pages = {3089-3106},
    year = {2008},
    month = {07},
    issn = {1364-5021},
    doi = {10.1098/rspa.2008.0189},
    url = {https://doi.org/10.1098/rspa.2008.0189},
    eprint = {https://royalsocietypublishing.org/rspa/article-pdf/464/2100/3089/742753/rspa.2008.0189.pdf},
}

@article{matchgate,
author = {Valiant, Leslie G.},
title = "{Quantum circuits that can be simulated classically in polynomial time}",
journal = {SIAM Journal on Computing},
volume = {31},
number = {4},
pages = {1229-1254},
year = {2002},
doi = {10.1137/S0097539700377025},
URL = {https://doi.org/10.1137/S0097539700377025},
eprint = {https://doi.org/10.1137/S0097539700377025}
}

@misc{kingma2017adam,
title = "{Adam: A Method for Stochastic Optimization}",
author = {Diederik P. Kingma and Jimmy Ba},
year = {2017},
eprint = {1412.6980},
archivePrefix = {arXiv},
primaryClass = {cs.LG},
url = {https://arxiv.org/abs/1412.6980}
}


\appendix

\section{Matchgate quantum circuit and group module}
\label{sec:app-matchgate-group-module}
\subsection{Matchgate circuit structure}

The variational ansatz depicted in Fig.~\ref{fig:mat-circuit} comprises \(O(n)\) circuit blocks. Each block contains three successive layers of quantum gates. 

We denote a gate index by
\[
q = (\mathfrak{b},\ell ,j),
\]
where \(\mathfrak{b}\) identifies the block, \(\ell\) the layer within that block, and \(j\) the qubit (or qubit pair) on which the gate acts.

In the first layer a single‑qubit \(Z\)–rotation is applied to every qubit \(j\in\{1,\dots ,n\}\):
\[
U(\theta)=\exp\!\bigl(-\tfrac{i}{2}\,\theta_q\,Z_j\bigr),
\qquad
\theta_q\in\mathbb{R}.
\]
The angles \(\{\theta_q\}_q\) are drawn from a random initialization and serve as the only variational parameters that are updated by the chosen optimization algorithm.  

The second and third layers contain two‑qubit gates acting on each adjacent pair of qubits \((i,i+1)\) with \(i=1,\dots ,n-1\). In these layer, the gate on the pair \((i,i+1)\) is
\begin{align}
U\bigl(\theta_q \bigr)
 =\exp\!\bigl(i\,\theta_q \,\gamma_{i}\bigr),
\end{align}
where the generator \(\gamma_{i}\) is independently sampled for each block from the set
\begin{align}
\{\,X_i\!\otimes\!X_{i+1},\;X_i\!\otimes\!Y_{i+1},\;Y_i\!\otimes\!X_{i+1},\;Y_i\!\otimes\!Y_{i+1}\,\}.
\end{align}
Thus, while the geometric placement of the two‑qubit gates (the pair \((i,i+1)\)) is identical for all blocks, the specific Pauli structure of each gate varies randomly from block to block. Once a block has been instantiated, the choice of generators remains fixed throughout the optimization.

\begin{figure*}[t]
    \centering
    \includegraphics[width=0.7\textwidth]{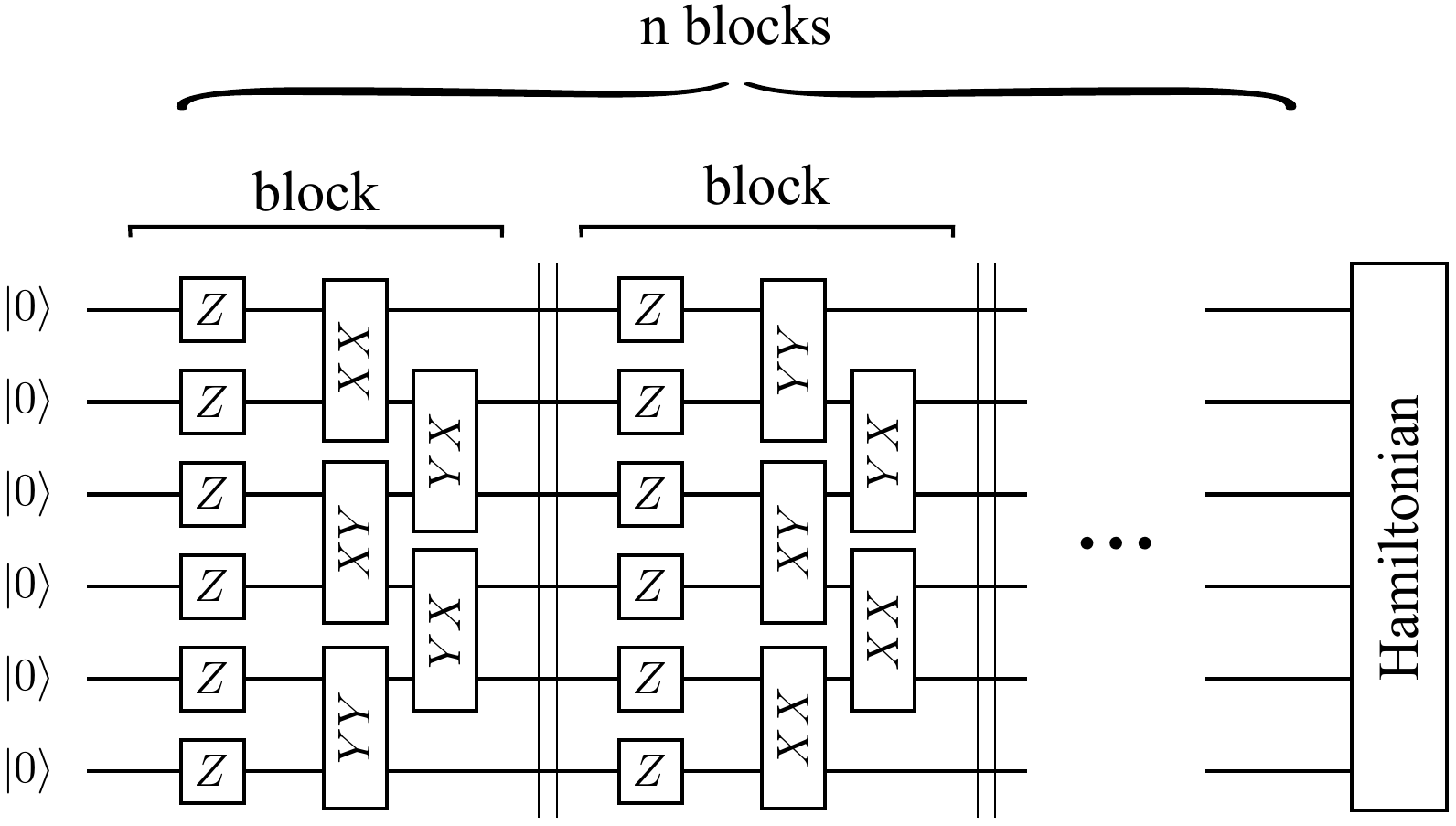}
    \caption{
    Schematic of the matchgate circuit employed in the optimization simulations. For an $n$‑qubit system the circuit consists of $n$ blocks; each block contains $n$ single‑qubit gates and $n-1$ two‑qubit gates acting on adjacent qubit pairs.
    }
    \label{fig:mat-circuit}
\end{figure*}

In summary, each block contains
\[
\underbrace{n}_{\text{single‑qubit}}+
\underbrace{n-1}_{\text{two‑qubit (layers 2 and 3)}}=2n-1
\]
variational parameters. Because the circuit comprises \(O(n)\) such blocks, the total number of independent variational parameters scales as \(O(n^{2})\).

The expressive power of the dynamical Lie algebra (DLA) is quantified by the dimension of the tangent space of admissible quantum states. Full expressivity is achieved when the number of variational parameters exceeds a critical value,
\[
N^{\text{(crit)}}\;\le\;\dim (\mathfrak{g}),
\]
below which training is challenging and above which optimization is facile~\cite{gsim,Larocca2023overpara}. For the matchgate‑generated DLA we have
\[
\dim (\mathfrak{g}) = O(n^{2}),
\]
which implies that the ansatz is over‑parameterized (\(N > N^{\text{(crit)}}\)). Consequently, the number of parameters is sufficient to achieve maximal expressivity within the matchgate DLA. In this regime, spurious local minima are absent, the loss landscape is markedly smoother, and variational optimization becomes substantially more tractable.

Furthermore, the circuit depth is sufficient to connect all computational‑basis states of the same parity. Since both the initial state and the target state of the MaxCut Hamiltonian reside in the computational basis, the chosen circuit architecture is well‑suited for the optimization task.

\subsection{Majorana operators and \(\kappa=4\) group module}
\label{appsec:majorana}

Matchgate circuits admit a fermionic representation in terms of \(2n\) Majorana operators constructed from Pauli strings.  A Pauli string on an \(n\)-qubit Hilbert space \(\mathcal{H}\) is defined as
\[
P=\bigotimes_{i=1}^{n} P_i,
\qquad
P_i\in\{I_i,X_i,Y_i,Z_i\}.
\]

For an \(n\)-qubit system, the Majorana operators are given by
\begin{align}
c^{X}_{1}&= X_{1}\,I_{2}\cdots I_{n},
&
c^{X}_{2}&= Z_{1}\,X_{2}\,I_{3}\cdots I_{n},
&
\dots~~
&
c^{X}_{n}= Z_{1}\cdots Z_{n-1}\,X_{n},
\nonumber\\[4pt]
c^{Y}_{1}&= Y_{1}\,I_{2}\cdots I_{n},
&
c^{Y}_{2}&= Z_{1}\,Y_{2}\,I_{3}\cdots I_{n},
&
\dots~~
&
c^{Y}_{n}= Z_{1}\cdots Z_{n-1}\,Y_{n}.
\label{eq:app-majorana}
\end{align}
They satisfy the Clifford algebra
\[
\{c^{\alpha}_{\mu},c^{\beta}_{\nu}\}
=
2\,\delta_{\mu\nu}\,\delta_{\alpha\beta},
\qquad
\alpha,\beta\in\{X,Y\},
\quad
\mu,\nu=1,\dots,n .
\]

Products of distinct Majorana operators induce a graded decomposition of the operator space.  The \(\kappa\)-th group module is defined as
\begin{align}
\mathcal{B}_{\kappa}
=
\operatorname{span}_{\mathbb{R}}
\left\{
c_{i_1}c_{i_2}\cdots c_{i_\kappa}
\,\middle|\,
1\le i_1<\dots<i_\kappa\le 2n
\right\},
\end{align}
with dimension \(d_{\kappa}=\binom{2n}{\kappa}\).

The module \(\mathcal{B}_{4}\) consists of all products of four distinct Majorana operators.  Each element reduces, up to a global phase, to a Pauli string.

Let \(\tilde b\) denote an unnormalised Pauli string.  Since Pauli strings are unitary and satisfy \(P_i^2=I_i\),
\[
\tilde b^{\dagger}\tilde b=\mathbb{I}.
\]
Therefore,
\[
\langle \tilde b,\tilde b\rangle_{\mathrm{HS}}
=
\operatorname{Tr}(\mathbb{I})
=
2^{n}.
\]
The normalised basis element is
\[
b=\frac{1}{\sqrt{2^{n}}}\,\tilde b,
\]
which satisfies \(\langle b,b\rangle_{\mathrm{HS}}=1\).

Multiplication of Pauli strings is carried out site by site using the single-qubit algebra, and any generated phase factors are absorbed into the global coefficient.  This ensures closure of the graded decomposition under operator multiplication.

\subsection{Connectivity of computational basis states under matchgate circuits}

A matchgate circuit can, in principle, connect any two computational basis states that share the same fermionic parity. Moreover, a circuit depth scaling as \(\Theta(n)\) is sufficient to realize this connectivity.

Nearest neighbor matchgate circuits preserve fermionic parity. Consequently, the Hilbert space decomposes into two invariant sectors corresponding to even and odd Hamming weight. This statement concerns the reachability of computational-basis states and does not imply that arbitrary superpositions within a fixed parity sector can be generated, as the accessible states are restricted to fermionic Gaussian states.

A gate included in matchgate circuit acts independently on the even-parity subspace and the odd-parity subspace, thereby conserving parity throughout the circuit. Such circuits are equivalent to fermionic linear optics and include generators corresponding to both hopping and pairing processes.

These operations suffice to connect all computational basis states within a fixed parity sector. Consider an arbitrary bit string \(x\) with even parity. By successive application of nearest neighbor hopping operations, excitations can be rearranged into a contiguous block, yielding a configuration of the form
\[
11\cdots 1100\cdots 0.
\]
Subsequently, pairing operations can be applied to each adjacent pair \(11 \to 00\), reducing the state to the vacuum \(\ket{00\cdots 0}\). Reversing this sequence enables the preparation of any other even-parity configuration \(y\). Hence, all even parity computational basis states are mutually connected.

An analogous construction applies to the odd-parity sector. After grouping excitations using hopping operations, repeated pair annihilation leaves a single excitation, producing a canonical representative such as \(\ket{10\cdots 0}\).
Reversing the procedure allows one to reach any other odd-parity configuration.
Therefore, within each parity sector, the set of computational basis states forms a single connected component under matchgate circuits.

In a one-dimensional nearest-neighbor architecture, the circuit depth required to connect arbitrary computational-basis states within a fixed parity sector scales linearly with system size. An upper bound of \(O(n)\) follows from a constructive procedure: nearest-neighbor hopping gates arranged in a brickwork pattern redistribute excitations across the chain in linear depth, while pairing operations can be applied in parallel to eliminate or create adjacent pairs.

\subsection{Effective Unitary Representation}

Let \(U(\boldsymbol{\theta})=\prod_q U(\theta_q)\) be a matchgate circuit, where each elementary unitary is generated by \(\gamma_q\in\mathfrak{g}\),
\begin{align}
U(\theta_q)=\exp(i\theta_q \gamma_q).
\end{align}
The effective representation \(\mathcal{M}(\theta_q)\) for each unitary gate can be achieved using Eq.\,\eqref{eq:eff-U}. The effective matrix for the full circuit is obtained as an ordered product,
\begin{align}
\mathcal{M}(\boldsymbol{\theta})
=
\prod_q \mathcal{M}(\theta_q).
\end{align}

Efficient computation of \(\mathcal{M}(\boldsymbol{\theta})\) follows from three algebraic properties of matchgate generators. First, each generator squares to identity,
\begin{align}
\gamma_q^2=I.
\label{eq:generator-cond-1}
\end{align}
Second, the commutator between a generator and a basis element is either zero or proportional to another basis element,
\begin{align}
\left[\gamma_q,b^{\,l}_{\kappa_p}\right]
=
\pm i\, b^{\,l'}_{\kappa_p}
\quad\text{or}\quad 0.
\label{eq:generator-cond-2}
\end{align}
Third, conjugation by a generator either preserves or flips the sign of a basis element,
\begin{align}
\gamma_q b^{\,l}_{\kappa_p} \gamma_q
=
\begin{cases}
\;\; b^{\,l}_{\kappa_p}, & [\gamma_q,b^{\,l}_{\kappa_p}]=0,\\
-\, b^{\,l}_{\kappa_p}, & [\gamma_q,b^{\,l}_{\kappa_p}]\neq 0.
\end{cases}
\label{eq:generator-cond-3}
\end{align}
These properties arise from the Pauli-string structure of both generators and basis elements.

Using Eq.~\eqref{eq:generator-cond-1}, the unitary can be written as
\begin{align}
U(\theta_q)
= \cos\theta_q\, I + i\sin\theta_q\, \gamma_q,
\nonumber
\end{align}
which yields
\begin{align}
U(\theta_q)\, b^{\,l}_{\kappa_p}\, U^{\dagger}(\theta_q)
&=
\cos^{2}\theta_q\, b^{\,l}_{\kappa_p}
+ i\sin\theta_q\cos\theta_q\,\bigl[\gamma_q, b^{\,l}_{\kappa_p}\bigr]
+ \sin^{2}\theta_q\, \gamma_q b^{\,l}_{\kappa_p} \gamma_q.
\nonumber 
\end{align}
Applying Eqs.~\eqref{eq:generator-cond-2}–\eqref{eq:generator-cond-3}, this expression simplifies to
\begin{align}
U(\theta_q)\, b^{\,l}_{\kappa_p}\, U^{\dagger}(\theta_q)
&= b^{\,l}_{\kappa_p},
&& \text{if }[\gamma_q,b^{\,l}_{\kappa_p}]=0,
\nonumber\\
&=
\cos(2\theta_q)\, b^{\,l}_{\kappa_p}
\mp
\sin(2\theta_q)\, b^{\,l'}_{\kappa_p},
&& \text{otherwise}.
\label{eq:eff-U-mat}
\end{align}

Therefore, each row of \(\mathcal{M}(\theta_q)\) contains at most two non-zero entries. The action of each gate is either trivial or a two-dimensional rotation within a pair of basis elements, resulting in a highly sparse effective representation.

\subsection{Group Module Distribution of Initial State}

Consider a zero-entanglement \(n\)-qubit pure product state.  Its density matrix factorises as
\[
\rho=\rho_1\otimes\rho_2\otimes\cdots\otimes\rho_n,
\]
where each \(\rho_i\) is a single-qubit pure state.  Every single-qubit density operator admits the Bloch representation
\begin{align}
\rho_i = \frac{1}{2} \left( I_i + a_i^X X_i + a_i^Y Y_i + a_i^Z Z_i \right),
\nonumber
\end{align}
with real coefficients \(a_i^{X,Y,Z}\in\mathbb{R}\) satisfying
\[
(a_i^X)^2+(a_i^Y)^2+(a_i^Z)^2=1.
\]
The global state is normalised as \(\operatorname{Tr}[\rho]=1\).

For computational-basis states \(|0\rangle\) and \(|1\rangle\), one has
\begin{align}
\rho_i^{|0\rangle}
=
\frac{1}{2}(I_i+Z_i),
\qquad
\rho_i^{|1\rangle}
=
\frac{1}{2}(I_i-Z_i). \nonumber
\end{align}
Hence an \(n\)-qubit computational-basis product state takes the form
\begin{align}
\rho
=
\frac{1}{2^n}
\prod_{i=1}^{n}
(I_i \pm Z_i),
\end{align}
which expands into a sum of Pauli strings containing only \(I_i\) and \(Z_i\).  Each term in this expansion belongs to a definite group module \(\mathcal{B}_{\kappa}\), where \(\kappa\) equals the number of non-identity factors in the corresponding Pauli string.

Although the relative signs depend on the specific bitstring, the absolute magnitude of every coefficient is fixed to \(2^{-n}\).  Consequently, all zero-entanglement computational-basis pure states share the same distribution of weights over the group modules \(\mathcal{B}_{\kappa}\).

\subsection{Computational Cost}
\label{appsec:complexity}
The expectation value in Eq.~\eqref{eq:eff-expectation} is obtained by successive applications of the matrices \(\mathcal{M}(\theta_q)\) to a coefficient vector, followed by an inner product with the Hamiltonian vector.

For the MaxCut Hamiltonian under matchgate dynamics, the relevant module has dimension
\[
d_{\kappa_p}= \dim (\mathcal{B}_{\kappa_p}) = O(n^4) .
\]
A naïve matrix–vector multiplication therefore requires \(O(n^8)\) operations per gate. Because \(O(n^2)\) gates are needed for sufficient expressivity of the matchgate DLA, the total cost of a single circuit evaluation scales as \(O(n^{10})\).

This cost can be dramatically reduced by exploiting the sparse structure of Eq.~\eqref{eq:eff-U-mat}. Rather than storing full matrices, one tracks how each generator permutes the basis elements.

For each generator \(\gamma_q\), the action on the basis \(\{b^{\,l}_{\kappa_p}\}_l\) is either trivial or maps a basis element to a single partner element. Since the number of distinct generators scales as \(O(n)\) and each acts on \(O(n^4)\) basis elements, the complete commutation structure can be stored using \(O(n^5)\) data.

This representation permits element‑wise application of each gate. The procedure is: (i) check whether the basis element commutes with the generator; if it does, the coefficient remains unchanged; (ii) if it does not commute, update the coefficient according to Eq.~\eqref{eq:eff-U-mat}. The number of basis elements that fail to commute with a given circuit generator scales as \(O(n^3)\), yielding a per‑gate computational complexity of \(O(n^3)\). Consequently, the total computational complexity for evaluating the entire circuit once scales as \(O(n^5)\).

\subsection{Experimental details}
\label{appsec:experiment}
The experiment is performed on a system equipped with NVIDIA RTX Pro 6000 Blackwell Max-Q GPU. Parameter optimization is performed using the Adam optimizer with a learning rate subject to a gradual decay schedule.
The optimization is terminated when the loss function fails to improve by more than \(10^{-5}\) over 50 consecutive iterations.

The optimization procedure for training the matchgate circuit ansatz follows a classical variational quantum eigensolver framework, adapted to leverage the special structure of matchgate dynamics. Parameters are optimized using gradient-based methods with automatic differentiation, where gradients of the expectation value $\langle H \rangle(\boldsymbol{\theta})$ are computed via back-propagation.

The gradient of the expectation value with respect to the variational parameters $\boldsymbol{\theta} = \{\theta_q\}$ is computed using automatic differentiation provided by PyTorch. For each parameter $\theta_q$, the gradient $\partial \langle H \rangle / \partial \theta_q$ is computed during the backward pass. The implementation supports gradient checkpointing to reduce memory consumption: instead of storing all intermediate activation values during the forward pass, certain gates are re-computed on-the-fly during back-propagation. This enables handling larger system sizes $n$ on GPU memory-limited environments.

Parameter optimization is performed using the Adam (Adaptive Moment Estimation) optimizer~\cite{kingma2017adam}, initialized with a learning rate of $\eta_{\text{initial}} = 0.05$. The optimizer maintains running estimates of the first and second moments of the gradients, providing adaptive learning rates for each parameter.

To facilitate convergence and escape local minima, a gradual learning rate decay schedule is employed. The algorithm monitors the loss function over consecutive iterations; if no improvement exceeding $\delta_{\text{threshold}} = 10^{-5}$ is observed for $P = 50$ iterations (after a minimum of $100$ steps), the learning rate is multiplied by a decay factor $f_{\text{decay}} = 0.5$. This process repeats until the learning rate reaches the minimum value $\eta_{\min} = 0.001$, at which point the optimization terminates. The decay sequence proceeds as:
\[
\eta: 0.05 \to 0.025 \to 0.0125 \to 0.00625 \to 0.003125 \to 0.015625 \to 0.001\text{ (min)}.
\]
This yields a maximum of five learning rate decay events.

The execution pipeline for a single optimization iteration proceeds as follows:

\begin{enumerate}
\item \textbf{State initialization}: Construct the initial state $\ket{\psi^0}$ (either all-zero $\ket{0}^{\otimes n}$ or single-qubit flipped $X_1\ket{0}^{\otimes n}$) and compute its representation in the Majorana basis.
\item \textbf{Forward pass}: Apply the matchgate circuit $U(\boldsymbol{\theta}) = \prod_q U(\theta_q)$ to the initial state by sequentially applying each gate's effective matrix $\mathcal{M}(\theta_q)$ to the coefficient vector.
\item \textbf{Expectation value}: Compute the inner product $\langle H \rangle = \bra{\psi(\boldsymbol{\theta})} H \ket{\psi(\boldsymbol{\theta})}$ using the evolved state coefficients and the Hamiltonian representation.
\item \textbf{Backward pass}: Perform automatic differentiation to compute gradients $\nabla_{\boldsymbol{\theta}} \langle H \rangle$, with optional gradient checkpointing.
\item \textbf{Parameter update}: Apply Adam optimizer to update $\boldsymbol{\theta} \leftarrow \boldsymbol{\theta} - \eta \cdot \text{Adam}(\nabla_{\boldsymbol{\theta}} \langle H \rangle)$.
\item \textbf{Check convergence}: Evaluate the improvement criterion; trigger learning rate decay or termination if conditions are met.
\end{enumerate}

\subsection{Convergence curves}
\label{appsec:convergence}

\begin{figure*}[htbp]
    \centering
    \includegraphics[width=0.95\textwidth]{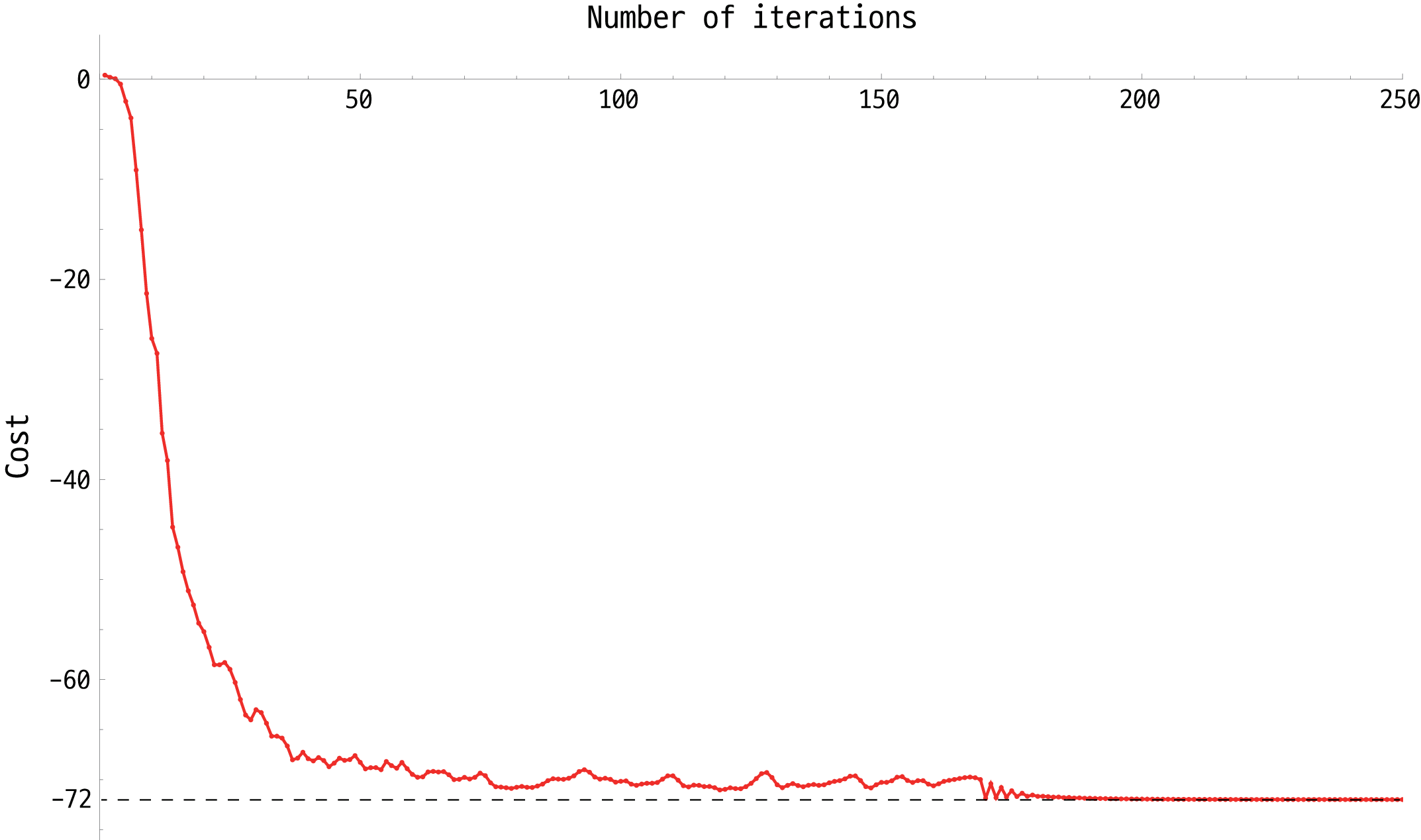}
    \caption{
    Convergence of the cost for an $n=60$ instance that the projected method reached the exact ground state.
    The lowest eigenvalue of the instance is $-72$.
    }
    \label{fig:conv_n=60}
\end{figure*}

\begin{figure*}[htbp]
    \centering
    \includegraphics[width=0.95\textwidth]{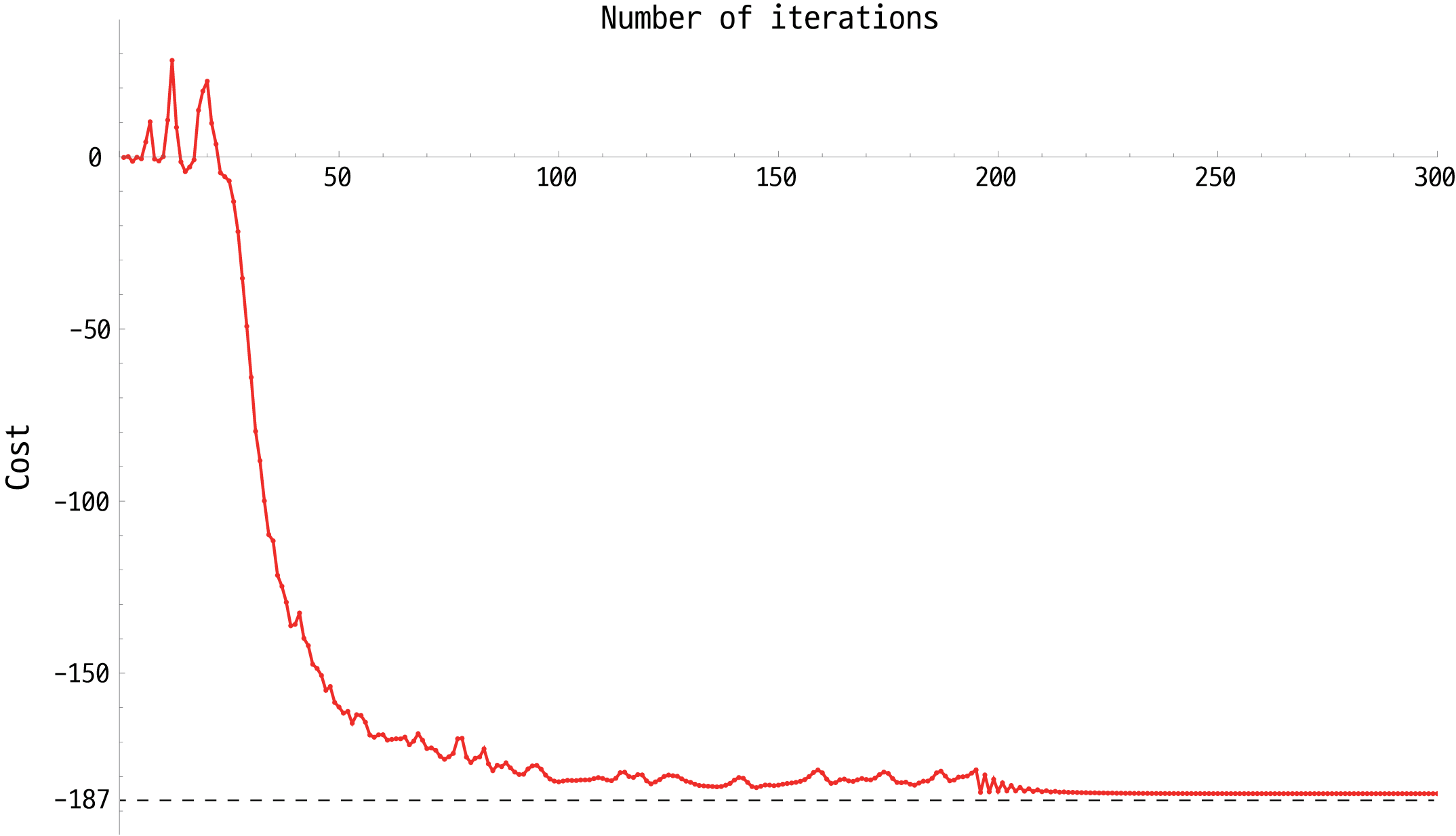}
    \caption{
    Convergence of the cost for $n=60$, $p=0.5$ Erd\H{o}s-R\'{e}nyi random graph instance \texttt{g05\_60.0} from Biq Mac Library~\cite{biqmac}.
    Out of nine trials, the best one we obtained was near $-185$ corresponding to the cut size of $535$, whereas the optimal cut is $536$ corresponding to $-187$ in cost.
    }
    \label{fig:conv_g05_60_0}
\end{figure*}

Figure~\ref{fig:conv_n=60} illustrates a convergence curve of an $n=60$ random 3-regular graph instance in which the exact ground state is obtained by the projected method.
For most other instances however, the algorithm failed as shown in Table~1 in the main text.

Figure~\ref{fig:conv_g05_60_0} shows a convergence of the cost for $n=60$, $p=0.5$ Erd\H{o}s-R\'{e}nyi random graph instance \texttt{g05\_60.0} from Biq Mac Library~\cite{biqmac}, where $p$ is the probability that an edge exists between two vertices.
Out of nine trials, the lowest cost it achieved was $-185$ corresponding to the cut size of $535$.
The maximum cut is $536$~\cite{biqmac}.

\section{Overparameterization in quantum circuits from a Fisher information perspective}

Fisher information quantifies the sensitivity of a probability distribution with respect to variations in a parameter. In the quantum setting, the Quantum Fisher Information (QFI) extends this notion by optimizing over all possible measurements, yielding a quantity that is intrinsic to the quantum state. The corresponding Quantum Fisher Information Matrix (QFIM) defines a Riemannian metric on the manifold of quantum states and characterizes the local geometry induced by parameter variations.

For a classical probability distribution \(p(x|\theta)\), the Fisher information is defined as
\[
F(\theta)
=
\int dx \, p(x|\theta)\big(\partial_\theta \log p(x|\theta)\big)^2.
\]
This quantity measures the distinguishability of nearby distributions under infinitesimal changes in \(\theta\).

In the quantum setting, for a parameterized state \(\rho(\theta)\), the QFI is defined as
\[
F_Q(\theta)
=
\max_{\mathrm{POVM}} F_{\mathrm{classical}}(\theta),
\]
which removes measurement dependence and promotes the quantity to an intrinsic property of the state.

For a multi-parameter family \(\boldsymbol{\theta} = (\theta_1,\dots,\theta_M)\), the QFIM for pure states is given by
\[
F_{ij}
=
4\,\mathrm{Re}\!\Big(
\langle \partial_i \psi | \partial_j \psi \rangle
-
\langle \partial_i \psi | \psi \rangle
\langle \psi | \partial_j \psi \rangle
\Big).
\]
This matrix induces a local metric via
\[
ds^2 = \frac{1}{2}\,\delta^T F \delta,
\]
thereby quantifying how infinitesimal parameter variations translate into distinguishable changes in the quantum state.

The rank of the QFIM provides a direct characterization of local expressivity. Since the QFIM is a Gram matrix of state derivatives, its rank equals the number of linearly independent directions in which the quantum state can change under infinitesimal parameter variations. Consequently, even in the presence of a large number of parameters, the effective degrees of freedom may be limited if parameter-induced variations are redundant.

The set of generators defining a parameterized quantum circuit induces a dynamical Lie algebra \(\mathfrak{g}\), which constrains the set of reachable unitaries and the corresponding state manifold. The rank of the QFIM is bounded by the dimension of this algebra,
\begin{align}
\operatorname{rank}(F) \le \dim(\mathfrak{g}_S),
\end{align}
where \(\mathfrak{g}_S\) denotes the effective algebra acting on the state. This bound reflects a fundamental limitation: no parameterization can generate more independent directions than those allowed by the underlying algebraic structure.

Overparameterization is characterized by the saturation of this bound. As the number of parameters \(M\) increases, the rank of the QFIM grows until it reaches its maximal value determined by \(\dim(\mathfrak{g}_S)\). Denoting by \(M_c\) the minimal number of parameters required for saturation, the regimes \(M < M_c\) and \(M \ge M_c\) correspond to underparameterized and overparameterized circuits, respectively. In typical settings, one observes that \(M_c\) scales with \(\dim(\mathfrak{g}_S)\).

This transition has direct implications for optimization. In the underparameterized regime, the mapping from parameter space to state space restricts the accessible directions, so that descent directions present in the full state space may not be representable within the chosen parameterization. This restriction leads to the emergence of spurious local minima. Once the QFIM rank saturates, all directions permitted by the dynamical Lie algebra become accessible. As a result, directions that were previously inaccessible are incorporated, and configurations that appeared as local minima in the restricted parameterization are revealed as saddle points in the full space.

Although the QFI is defined through an optimization over measurements, the QFIM is employed here as a geometric object independent of any specific observable. In practical optimization problems involving the expectation value of a fixed observable \(O\), the QFIM does not guarantee favorable behavior for every choice of \(O\). However, it determines whether the parameterization spans all physically accessible directions in state space. If a descent direction exists but lies outside the accessible subspace, optimization necessarily fails. Conversely, once the accessible subspace coincides with the full Lie-algebraic orbit, such structural obstructions are removed.

The essential conclusion is that the trainability of parameterized quantum circuits is governed by the number of independent directions they induce in state space, rather than by the raw number of parameters. This quantity is captured by the rank of the QFIM and is fundamentally limited by the dimension of the dynamical Lie algebra. When this bound is saturated, the circuit achieves maximal local expressivity, and the associated removal of spurious local minima leads to significantly improved optimization performance.

\section{Classical simulation of quantum circuits via group module restricted representations}
\label{sec:app-classical-sim}

This note summarizes key elements of the framework introduced in \emph{Lie-algebraic classical simulation for quantum computing} (Matthew L. Goh \emph{et al.}, PRR 2025)~\cite{gsim}, focusing on the conditions under which expectation values can be computed efficiently.

When the observable operator is restricted to a structured subspace of the operator algebra, its expectation value with respect to an arbitrary quantum state can be evaluated in polynomial time.

A \emph{special operator} is defined as an observable fully supported on a subspace \(\mathcal{L}_{\lambda}\), where
\[
\mathcal{L}_{\lambda} = \bigoplus_{\kappa'} \mathcal{B}_{\kappa'} \, ,
\]
and each group module \(\mathcal{B}_{\kappa'}\) satisfies
\[
\dim (\mathcal{B}_{\kappa'}) \le \operatorname{poly}(n) \, .
\]
The restriction to such subspaces ensures that the relevant degrees of freedom scale polynomially with the system size.

The action of a unitary operator on a special operator is confined to \(\mathcal{L}_{\lambda}\). Consequently, the operator dynamics can be represented as a linear transformation acting on a vector space of dimension \(\dim (\mathcal{L}_{\lambda})\), enabling an efficient classical description.

For a general quantum state, only the components supported on \(\mathcal{L}_{\lambda}\) contribute to the expectation value of a special operator. As a result, the expectation value reduces to an inner product between the vector representation of the state restricted to \(\mathcal{L}_{\lambda}\) and the vector representation of the transformed observable.

The dominant computational cost arises from applying the matrix representation of the unitary transformation within this reduced space, which scales as
\[
O\!\big(\dim(\mathcal{L}_{\lambda})^{2}\big) \, ,
\]
and is therefore polynomial in \(n\).

We consider a parametrized quantum circuit of the form
\[
U(\boldsymbol{\theta}) = \prod_{q=1}^{Q} e^{i\theta_q \gamma_q} \, ,
\]
where \(\boldsymbol{\theta} = [\theta_1,\dots,\theta_Q]\) and the generators satisfy \(\gamma_q \in \mathfrak{g}\).

The associated dynamical Lie algebra is defined as
\[
\mathfrak{g} = \operatorname{span}_{\mathbb{R}} \!\big\langle \{ i\gamma_q \}_{q=1}^{Q} \big\rangle \subseteq \mathfrak{su}(2^n) \, ,
\]
i.e., the real vector space generated by all nested commutators of \(\{ i\gamma_q \}\).

The corresponding dynamical Lie group is obtained via exponentiation,
\[
\mathcal{U} = e^{\mathfrak{g}} \equiv \{ e^{iA} \mid iA \in \mathfrak{g} \} \, ,
\]
and consists of all unitaries generated by elements of the algebra.

Let \(\mathcal{L}_{\lambda} \subset \mathcal{B}\) denote a subspace of the operator algebra decomposed into group modules as above. For any \(iA \in \mathfrak{g}\) and basis element \(b_{\kappa'}^{\,l} \in \mathcal{B}_{\kappa'}\), the adjoint action is given by
\[
\left[ A, b_{\kappa'}^{\,l} \right]
= \sum_{l'=1}^{\dim(\mathcal{B}_{\kappa'})} i\, \bar{A}_{l,l'} \, b_{\kappa'}^{\,l'} \, ,
\]
with matrix elements
\[
\bar{A}_{l,l'} = -i\, \operatorname{Tr}\!\big[\, b_{\kappa'}^{\,l'} [A, b_{\kappa'}^{\,l}] \,\big] \, .
\]
The adjoint representation preserves the module decomposition, yielding a block-diagonal structure
\[
\bar{A} = \bigoplus_{\kappa'} \bar{A}_{\kappa'} \, .
\]

For a unitary \(U = e^{iA}\), the transformation of the basis elements becomes
\[
U\, b_{\kappa'}^{\,l} \, U^{\dagger}
= \sum_{l'} \big[ e^{i\bar{A}_{\kappa'}} \big]_{l,l'} \, b_{\kappa'}^{\,l'} \, .
\]
Accordingly,
\[
e^{i\bar{A}} = \bigoplus_{\kappa'} e^{i\bar{A}_{\kappa'}} \, .
\]

A special operator takes the form
\[
O = \sum_{b_{\kappa'}^{\,l} \in \mathcal{L}_{\lambda}}
\omega^{(\lambda)}_{\kappa',l} \, b_{\kappa'}^{\,l} \, ,
\]
where the coefficient vector \(\boldsymbol{\omega}^{(\lambda)}\) decomposes into module-wise components of polynomial size.

For a quantum state evolving as
\[
\rho^{\mathrm{out}} = U \rho^{\textrm{i}} U^{\dagger} \, ,
\]
only the projections of \(\rho^{\textrm{i}}\) onto the modules supporting \(O\) contribute. Define
\[
\varphi^{\textrm{i}}_{\kappa',l}
= \operatorname{Tr}\!\big[\, b_{\kappa'}^{\,l} \rho^{\textrm{i}} \,\big] \, .
\]
The transformed coefficients are
\[
\varphi^{\mathrm{out}}_{\kappa',l}
= \sum_{l'} \varphi^{\textrm{i}}_{\kappa',l'} \,
\big[ e^{i\bar{A}_{\kappa'}} \big]_{l',l} \, .
\]

The expectation value is then expressed as
\begin{align}
\langle O \rangle
= \operatorname{Tr}\!\big[ \rho^{\mathrm{out}} O \big]
= \boldsymbol{\varphi}^{\mathrm{out}} \!\cdot\! \boldsymbol{\omega}^{(\lambda)}
= \boldsymbol{\varphi}^{\textrm{i}} \!\cdot\!
\left( \bigoplus_{\kappa'} e^{i\bar{A}_{\kappa'}} \right)
\!\cdot\! \boldsymbol{\omega}^{(\lambda)} \, .
\end{align}

In a circuit consisting of \(L\) parametrized gates, each gate admits a block-diagonal representation over the modules. The cost of applying a single gate within a module scales as \(O\!\big(\dim(\mathcal{B}_{\kappa'})^{2}\big)\), leading to an overall complexity
\[
O\!\big( L \, \max_{\kappa'} \{ \dim(\mathcal{B}_{\kappa'}) \}^{2} \big) \, .
\]
Since \(\dim(\mathcal{B}_{\kappa'})\) grows polynomially with \(n\), the full simulation cost remains polynomial. Consequently, expectation values of special operators can be efficiently computed on a classical computer within this framework.

\section{Classical simulation of quantum circuits via Wick-based techniques}

This note summarizes key elements of the framework introduced in \emph{Lie-algebraic classical simulation for quantum computing} (Matthew L. Goh \emph{et al.}, PRR 2025), focusing on the role of highest-weight states and Wick-based methods.

For a given quantum circuit, there exists a class of states—referred to as \emph{special states}—for which classical simulation of expectation values can be performed efficiently.

The Lie algebra \(\mathfrak{g}\) describing the circuit admits a Cartan–Weyl decomposition into a Cartan subalgebra and sets of raising and lowering operators. The Cartan subalgebra \(\mathfrak{h}\) is the maximal abelian subalgebra consisting of mutually commuting elements, and therefore admits a simultaneous eigenbasis.

A highest-weight state is defined as a simultaneous eigenstate of all Cartan generators that is annihilated by all raising operators. Such highest-weight states constitute the special states relevant for efficient simulation.

When an observable belongs to the Lie algebra, it can be decomposed into components supported on the Cartan subalgebra and on the raising and lowering subspaces. Upon evaluating expectation values with respect to a highest-weight state, only the Cartan components contribute, while the remaining terms vanish due to the annihilation properties of the state.

Even when the observable does not lie within the Lie algebra, it can be expressed as a product of Lie-algebra elements. By decomposing each factor into Cartan, raising, and lowering components and systematically reordering operators using commutation relations, lowering operators can be moved to the left and raising operators to the right. The resulting expression contains only a small number of non-vanishing contributions, as all other terms annihilate the highest-weight state or its dual. Provided that the operator length remains moderate, this procedure can be carried out efficiently.

A representative example is given by matchgate quantum circuits. In this case, the highest-weight states correspond to computational-basis states.

We consider a family of quantum circuits generated by a Lie algebra \(\mathfrak{g}\), with unitary evolution
\[
U = e^{iH}, \qquad H \in \mathfrak{g}.
\]
The Lie algebra admits a Cartan–Weyl decomposition
\begin{align}
\mathfrak{g} = \mathfrak{h} \oplus \bigoplus_{\alpha} \mathfrak{g}_{\alpha} \oplus \mathfrak{g}_{-\alpha},
\end{align}
where \(\mathfrak{h}\) is the Cartan subalgebra and \(\mathfrak{g}_{\pm \alpha}\) are generated by raising and lowering operators \(E_{\pm \alpha}\), respectively.

A highest-weight state \(|\mathrm{hw}\rangle\) satisfies
\[
E_{\alpha} |\mathrm{hw}\rangle = 0
\]
for all positive roots \(\alpha\), while remaining a simultaneous eigenstate of all elements of \(\mathfrak{h}\).

For an observable \(O \in \mathfrak{g}\), the expectation value after evolution is
\[
\langle O \rangle
= \langle \mathrm{hw} | U^{\dagger} O U | \mathrm{hw} \rangle.
\]
Since the adjoint action preserves the algebra,
\[
U^{\dagger} O U \in \mathfrak{g},
\]
and can be expanded as
\[
U^{\dagger} O U
= \sum_{k} c_{k}\,\tilde{H}_{k}
+ \sum_{\alpha} c_{\alpha}\,E_{\alpha}
+ \sum_{\alpha} c_{-\alpha}\,E_{-\alpha}.
\]
Using the highest-weight property, only the Cartan terms contribute, yielding
\[
\langle \mathrm{hw} | U^{\dagger} O U | \mathrm{hw} \rangle
= \sum_{k} c_{k}\,\tilde{h}_{k},
\]
where \(\tilde{h}_{k}\) denote the eigenvalues of the Cartan generators on \(|\mathrm{hw}\rangle\).

For observables expressed as products of Lie-algebra elements,
\[
O = O_{1} O_{2}, \qquad O_{1}, O_{2} \in \mathfrak{g},
\]
one has
\[
U^{\dagger} O U
= (U^{\dagger} O_{1} U)(U^{\dagger} O_{2} U).
\]
Each factor is expanded in the Cartan–Weyl basis, and the operators are reordered using commutation relations so that lowering operators are positioned to the left and raising operators to the right. The annihilation properties of the highest-weight state ensure that only a restricted set of terms survive, leading to an efficient evaluation.

Matchgate circuits provide a concrete realization of this structure. These circuits correspond to free-fermionic evolutions and admit a fermionic representation in terms of creation and annihilation operators. The fermionic vacuum
\[
a_{j} |0\rangle = 0
\]
serves as a highest-weight state, which corresponds in the qubit representation to the computational-basis state
\[
|0\,0\,\cdots\,0\rangle.
\]

In practice, classical simulation of matchgate circuits is performed using the fermionic Gaussian formalism rather than an explicit Cartan–Weyl decomposition. Introducing Majorana operators \(\{c_{a}\}_{a=1}^{2n}\), one defines the covariance matrix
\begin{align}
\Gamma_{ab} = \langle c_{a} c_{b} \rangle.
\end{align}
For fermionic Gaussian states, Wick's theorem implies
\[
\langle c_{a_{1}} \cdots c_{a_{2m}} \rangle
= \operatorname{Pf}\!\big( \Gamma_{a_{i} a_{j}} \big),
\]
so that higher-order correlators reduce to Pfaffians of two-point functions.

Under matchgate evolution, Majorana operators transform linearly,
\[
U^{\dagger} c_{a} U
= \sum_{b} R_{ab} c_{b}, \qquad R \in SO(2n),
\]
which induces the transformation
\begin{align}
\Gamma \;\mapsto\; R \Gamma R^{T}.
\end{align}
Because the Majorana operator space is closed under this transformation, the Gaussian structure of the state is preserved.

Consequently, matchgate circuits implement orthogonal transformations in Majorana space associated with the group \(SO(2n)\). Gaussian states remain Gaussian under the dynamics, and all correlation functions can be computed efficiently from the covariance matrix using Pfaffians. This structure underlies the efficient classical simulation of matchgate circuits acting on computational-basis states.

\end{document}